\documentclass[11pt]{article}

\usepackage{times}
\usepackage{latexsym}
\usepackage{amsfonts,amsthm,amssymb}
\usepackage{euscript}
\usepackage{amstext}
\usepackage{graphicx}
\usepackage{color}
\usepackage{url,hyperref}

\newtheorem{lemma}{Lemma}[section]
\newtheorem{theorem}[lemma]{Theorem}

\newtheorem{corollary}[lemma]{Corollary}

\newtheorem{remark}[lemma]{Remark}

\newtheorem{observation}{Observation}

        {\hspace*{\fill}$\Box$\par}

\newcommand{\lemref}[1]{Lemma~\ref{lemma:#1}}

\newcommand{\corref}[1]{Corollary~\ref{cor:#1}}

\newcommand{\lemlab}[1]{\label{lemma:#1}}

\newcommand{\corlab}[1]{\label{cor:#1}}

\newcommand{\opt}{\textrm{\sc OPT}}
\newcommand{\etal}{et al.\ }
\newcommand{\eps}{\epsilon}
\newcommand{\st}{S} 
\newcommand{\Algorithm}[1]{{\texttt{\bf{#1}}}} 
\newcommand{\sbg}{\Algorithm{SSF-W}} 
\newcommand{\sug}{\Algorithm{SSF}} 
\newcommand{\mmug}{\Algorithm{SSF-ID}} 

\begin{document}
\title{Online Scheduling to Minimize the Maximum Delay Factor}
\author{
Chandra Chekuri\thanks{Department of Computer Science,
University of Illinois, 201 N.\ Goodwin Ave.,
Urbana, IL 61801. {\tt chekuri@cs.uiuc.edu}.
Partially supported by NSF grants CCF 0728782
 and CNS 0721899. }
\and
Benjamin Moseley\thanks{Department of Computer Science,
University of Illinois, 201 N.\ Goodwin Ave.,
Urbana, IL 61801. {\tt bmosele2@uiuc.edu}.
}
}
\date{\today}
\maketitle

\begin{abstract}
  In this paper two scheduling models are addressed. First is the
  standard model (unicast) where requests (or jobs) are
  independent. The other is the broadcast model where broadcasting a
  page can satisfy multiple outstanding requests for that page. We
  consider online scheduling of requests when they have deadlines.
  Unlike previous models, which mainly consider the objective of
  maximizing throughput while respecting deadlines, here we focus on
  scheduling all the given requests with the goal of minimizing the
  maximum {\em delay factor}. The delay factor of a schedule is
  defined to be the minimum $\alpha \ge 1$ such that each request $i$
  is completed by time $a_i + \alpha(d_i - a_i)$ where $a_i$ is the
  arrival time of request $i$ and $d_i$ is its deadline. Delay factor
  generalizes the previously defined measure of maximum stretch which
  is based only the processing times of requests
  \cite{BenderCM98,BenderMR02}.

  We prove strong lower bounds on the achievable competitive ratios
  for delay factor scheduling even with unit-time requests. Motivated
  by this, we then consider resource augmentation analysis
  \cite{KalyanasundaramP95} and prove the following positive results.
  For the unicast model we give algorithms that are $(1 + \eps)$-speed
  $O({1 \over \eps})$-competitive in both the single machine and
  multiple machine settings. In the broadcast model we give an
  algorithm for similar-sized pages that is $(2+ \eps)$-speed $O({1
    \over \eps^2})$-competitive. For arbitrary page sizes we give an
  algorithm that is $(4+\eps)$-speed $O({1 \over
    \eps^2})$-competitive.
\end{abstract}

\setcounter{page}{0}
\thispagestyle{empty}
\clearpage

\section{Introduction}
Scheduling requests (or jobs\footnote{In this paper we use requests
  instead of jobs since we also address the broadcast scheduling
  problem where a request for a page is more appropriate terminology
  than a job.}) that arrive online is a fundamental problem faced by
many systems and consequently there is a vast literature on this
topic. A variety of models and performance metrics are studied in
order to capture the requirements of a system. In this work, we
consider a recently suggested performance measure called {\em delay
  factor} \cite{ChangEGK08} when each request has an {\em arrival
  time} (also referred to as release time) and a {\em deadline}. We
consider both the traditional setting where requests are independent,
and also the more recent setting of broadcast scheduling when
different requests may ask for the same page (or data) and can be
simultaneously satisfied by a single transmission of the page. We
first describe the traditional setting, which we refer to as the {\em
  unicast} setting, to illustrate the definitions and and then
describe the extension to the {\em broadcast} setting.

We assume that requests arrive {\em online}. The arrival time $a_i$,
the deadline $d_i$, and the processing time $\ell_i$ of a request
$J_i$ are known only when $i$ arrives. We refer to the quantity $S_i =
(d_i-a_i)$ as the {\em slack} of request $i$. There may be a single
machine or $m$ identical machines available to process the
requests. Consider an online scheduling algorithm $A$. Let $f_i$
denote the completion time or finish time of $J_i$ under $A$. Then the
delay factor of $A$ on a given sequence of requests $\sigma$ is
defined as $\alpha^A(\sigma) = \max\{1, \max_{J_i \in \sigma}
\frac{f_i - a_i}{d_i - a_i}\}$. In other words $\alpha^A$ measures the
factor by which $A$ has delayed jobs in {\em proportion} to their
slack. The goal of the
scheduler is to minimize the (maximum) delay factor. We consider
worst-case competitive analysis. An online algorithm $A$ is
$r$-competitive if for all request sequences $\sigma$,
$\alpha^A(\sigma) \le r \alpha^*(\sigma)$ where $\alpha^*(\sigma)$ is
the delay factor of an optimal offline algorithm.  Delay factor
generalizes the previously studied maximum stretch measure introduced
by Bender, Chakraborty and Muthukrishnan \cite{BenderCM98}. The
maximum stretch of a schedule $A$ is $\max_{J_i \in \sigma} (f_i -
a_i)/\ell_i$ where $\ell_i$ is the length or processing time of
$J_i$. By setting $d_i =a_i + \ell_i$ for each request $J_i$ it can be
seen that delay factor generalizes maximum stretch.

In the broadcast setting, multiple requests can be satisfied by the
same transmission. This model is inspired by a number of recent
applications --- see \cite{BarnoyBNS98, AksoyF98, AcharyaFZ95,
  BartalM00} for the motivating applications and the growing
literature on this topic. More formally, there are $n$ distinct pages
or pieces of data that are available in the system, and clients can
request a specific page at any time. This is called the {\em
  pull-model} since the clients initiate the request and we focus on
this model in this paper (in the push-model the server transmits the
pages according to some frequency). Multiple outstanding requests for
the same page are satisfied by a single transmission of the page. We
use $J_{(p,i)}$ to denote $i$'th request for a page $p \in \{1, 2,
\ldots, n\}$. We let $a_{(p,i)}$ and $d_{(p,i)}$ denote the arrival
time and deadline of the request $J_{(p,i)}$. The finish time
$f_{(p,i)}$ of a request $J_{(p,i)}$ is defined to be the earliest time after
$a_{(p,i)}$ when the page $p$ is sequentially transmitted by the
scheduler. Note that multiple requests for the same page can have the
same finish time. The delay factor $\alpha^A$ for an algorithm $A$ over
a sequence of requests $\sigma$ is now defined as $ \max\{1, \max_{(p,i) \in
  \sigma}\frac{f_{(p,i)} - a_{(p,i)}}{d_{(p,i)} - a_{(p,i)}}\}$.

\medskip
\noindent {\bf Motivation:} There are a variety of metrics in the
scheduling literature and some of the well-known and widely used ones
are makespan and average response time (or flowtime). More recently,
other metrics such as maximum and average stretch, which measure the
waiting time in proportion to the size of a request, have been
proposed \cite{BenderCM98, Karger99, Sgall98}; these measures were
motivated by applications in databases and web server systems. Related
metrics include $L_p$ norms of response times and stretch
\cite{BansalP03, AvrahamiA03, ChekuriGKK04} for $1 \le p < \infty$. In
a variety of applications such as real-time systems and data gathering
systems, requests have deadlines by which they desire to be
fulfilled. In real-time systems, a {\em hard} deadline implies that it
cannot be missed, while a {\em soft} deadline implies some flexibility
in violating it. In online settings it is difficult to respect hard
deadlines. Previous work has addressed hard deadlines by either
considering periodic tasks or other restrictions \cite{Burnsb08}, or
by focusing on maximizing throughput (the number of jobs completed by
their deadline) \cite{Kimc04, ChanLTW04, ZhengFCCPW06}. It was
recently suggested by Chang \etal \cite{ChangEGK08} that delay factor
is a useful and natural relaxation to consider in situations with soft
deadlines where we desire all requests to be satisfied. In addition,
as we mentioned already, delay factor generalizes maximum stretch
which has been previously motivated and studied in
\cite{BenderCM98,BenderMR02}.

\medskip
\noindent {\bf Results:} We give the first results for {\em online}
scheduling for minimizing delay factor in both the unicast and
broadcast settings. Throughout we assume that requests are allowed to
be {\em preempted} if they have varying processing times. We first
prove strong lower bounds on online competitiveness.
\begin{itemize}
\item For unicast setting no online algorithm is
  $\Delta^{0.4}/2$-competitive where $\Delta$ is the ratio between
  the maximum and minimum slacks.
\item For broadcast scheduling with $n$ unit-sized pages there is no
$n/4$-competitive algorithm.
\end{itemize}
We resort to resource augmentation analysis, introduced by of
Kalyanasundaram and Pruhs \cite{KalyanasundaramP95}, to overcome the
above lower bounds. In this analysis the online algorithm is given
faster machines than the optimal offline algorithm. For $s \ge 1$, an
algorithm $A$ is $s$-speed $r$-competitive if $A$ when given $s$-speed
machine(s) achieves a competitive ratio of $r$. We prove the
following.
\begin{itemize}
\item For unicast setting, for any $\eps \in (0,1]$, there are $(1+\eps)$-speed
  $O(1/\eps)$-competitive algorithms in both single and multiple machine
  cases. Moreover, the algorithm for the multiple machine case immediately
  dispatches an arriving request to a machine and is non-migratory.
\item For broadcast setting, for any $\eps \in (0,1]$, there is a
  $(2+\eps)$-speed $O(1/\eps^2)$-competitive algorithm for unit-sized
  (or similar sized) pages.  If pages can have varying length, then
  for any $\eps \in (0,1]$, there is a $(4+\eps)$-speed
  $O(1/\eps^2)$-competitive algorithm.
\end{itemize}

Our results for the unicast setting are related to, and borrow
ideas from, previous work on minimizing $L_p$ norms of response
time and stretch \cite{BansalP03} in the single machine and
parallel machine settings \cite{AvrahamiA03,ChekuriGKK04}.

Our main result is for broadcast scheduling. Broadcast scheduling has
posed considerable difficulties for algorithm design. In fact most of
the known results are for the {\em offline} setting
\cite{KalyanasundaramPV00,ErlebachH02,GandhiKKW04,GandhiKPS06,BansalCS06,BansalCKN05}
and several of these use resource augmentation! The difficulty in
broadcast scheduling arises from the fact that the online algorithm
may transmit a page multiple times to satisfy distinct requests for
the same page, while the offline optimum, which knows the sequence in
advance, can {\em save work} by gathering them into a single
transmission. Online algorithms that maximize throughput
\cite{Kimc04,ChanLTW04,ZhengFCCPW06,ChrobakDJKK06} get around this by
eliminating requests. Few positive results are known in the online
setting where all requests need to be scheduled
\cite{BartalM00,EdmondsP03,EdmondsP04} and the analysis in all of
these is quite non-trivial.
In contrast, our algorithm and analysis
are direct and explicitly demonstrate the value of making requests
wait for some duration so as to take advantage of potential future
requests for the same page. We hope this idea can be further exploited
in other broadcast scheduling contexts. We mention that even in the
{\em offline} setting, only an LP-based $2$-speed algorithm is known
for delay factor with unit-sized pages \cite{ChangEGK08}.

\medskip
\noindent {\bf Related Work:} We refer the reader to the survey on
online scheduling by Pruhs, Sgall and Torng \cite{PruhsST} for a
comprehensive overview of results and algorithms (see also
\cite{Pruhs07}). For jobs with deadlines, the well-known
earliest-deadline-first (EDF) algorithm can be used in the offline
setting to check if all the jobs can be completed before their
deadline. A substantial amount of literature exists in the real-time
systems community in understanding and characterizing restrictions on
the job sequence that allow for schedulability of jobs with deadlines
when they arrive online or periodically. Previous work on soft
deadlines is also concerned with characterizing inputs that allow for
bounded tardiness. We refer the reader to
\cite{RealtimeHandbook} for the extensive literature
scheduling issues in real-time systems.

Closely related to our work is that on max stretch \cite{BenderCM98}
where it is shown that no online algorithm is $O(P^{0.3})$ competitive
even in the preemptive setting where $P$ is ratio of the largest job
size to the smallest job size. \cite{BenderCM98} also gives an
$O(\sqrt{P})$ competitive algorithm which was further refined in
\cite{BenderMR02}. Resource augmentation analysis for $L_p$ norms of
response time and stretch from the work of Bansal and Pruhs
\cite{BansalP03} implicitly shows that the shortest job first (SJF)
algorithm is a $(1+\eps)$-speed $O(1/\eps)$-competitive algorithm for
max stretch. Our work shows that this analysis can be generalized for
the delay factor metric. For multiple processors our analysis is
inspired by the ideas from \cite{AvrahamiA03,ChekuriGKK04}.

Broadcast scheduling has seen a substantial amount of research in
recent years; apart from the work that we have already cited we refer
the reader to \cite{CharikarK06,KhullerK04}, the recent paper of Chang
\etal \cite{ChangEGK08}, and the surveys \cite{PruhsST,Pruhs07} for
several pointers to known results. Our work on delay factor is
inspired by \cite{ChangEGK08}.  As we mentioned already, a good
amount of the work on broadcast scheduling has been on offline
algorithms including NP-hardness results and approximation algorithms
(often with resource augmentation). For delay factor there is a
$2$-speed optimal algorithm in the {\em offline} setting and it is
also known that unless $P=NP$ there is no $2-\eps$ approximation
\cite{ChangEGK08}. In the online setting the following results are
known.  For maximum response time, it is shown in
\cite{BartalM00,ChangEGK08} that first-in-first-out (FIFO) is
$2$-competitive.  For average response time,
Edmonds and Pruhs \cite{EdmondsP03} give a $(4+\eps)$-speed
$O(1/\eps)$-competitive algorithm; their algorithm is an indirect
reduction to a complicated algorithm of Edmonds \cite{Edmonds00} for
non-clairvoyant scheduling. They also show in \cite{EdmondsP04} that
longest-wait-first (LWF) is a $6$-speed $O(1)$-competitive algorithm
for average response time.  Constant competitive online
algorithms for maximizing throughput
\cite{Kimc04,ChanLTW04,ZhengFCCPW06,ChrobakDJKK06} for unit-sized
pages.

\medskip
\noindent
We describe our results for the unicast setting
in Section~\ref{sec:unicast} and for the broadcast settings in
Section~\ref{sec:broadcast}.

\medskip
\noindent
{\bf Notation:}
We let $S_i = d_i - a_i$ denote the slack of $J_i$ in the unicast
setting. When requests have varying processing times (or lengths) we
use $\ell_i$ to denote the length of $J_i$. We assume without loss
of generality that $S_i \ge \ell_i$.
In the broadcast setting, $(p,i)$ denotes the $i$'th request for
page $p$. We assume that the requests for a page are ordered by time
and hence $a_{(p,i)} \le a_{(p,j)}$ for $i < j$.
In both settings we use $\Delta$ to denote the ratio of maximum slack
to the minimum slack in a given request sequence.

\section{Unicast Scheduling}
\label{sec:unicast}
In this section we address the unicast case where requests are
independent. We may thus view requests as jobs although we stick with
the use of requests. For a request $J_i$, recall that
$a_i,d_i,\ell_i,f_i$ denote the arrival time, deadline, length, and
finish time respectively. An instance with all $\ell_i = 1$ (or more
generally the processing times are the same) is referred to as a
unit-time instance. It is easy to see that preemption does not help
much for unit-sized instances. Assuming that the processing times are
integer valued then in the single machine setting one can reduce an
instance with varying processing time to an instance with unit-times
as follows. Replace $J_i$, with length $\ell_i$, by $\ell_i$
unit-sized requests with the same arrival and deadline as that of
$J_i$.

As we had remarked earlier, scheduling to minimize the maximum stretch
is a special case of scheduling to minimize the maximum delay factor.
In \cite{BenderCM98} a lower bound of $P^{1/3}$ is shown for online
maximum stretch on a $1$-speed machine where $P$ is the ratio of the
maximum processing time to the minimum processing time. They show that
this bounds holds even when $P$ is known to the algorithm.  This
implies a lower bound of $\Delta^{1/3}$ for minimizing the maximum delay
factor. Here we improve the lower bound for maximum stretch to
$P^{0.4}/2$ when the online algorithm is not aware of $P$. A proof
can be found in the appendix.

\begin{theorem}
  \label{thm:unicast_lb}
  There is no $1$-speed ${{P^{.4}} \over 2}$-competitive algorithm
  for online maximum stretch when $P$ is not known in advance to
  the algorithm.
\end{theorem}

\begin{corollary}
\label{delayonlylower}
There is no $1$-speed ${{\Delta^{.4}} \over 2}$-competitive algorithm
for delay factor scheduling when $\Delta$ is not known in advance with unit-time requests.
\end{corollary}

In the next two subsections we show that with $(1+\eps)$ resource
augmentation simple algorithms achieve an $O(1/\eps)$ competitive
ratio.

\subsection{Single Machine Scheduling}
\label{unicast-single}

We analyze the simple shortest-slack-first ($\sug$) algorithm
which at any time $t$ schedules the request with the shortest slack.

\begin{center}
\begin{tabular}[r]{|c|}
\hline
\textbf{Algorithm}: \sug \\

\begin{minipage}{13cm}
\begin{itemize}
\item At any time $t$ schedule the request with with the minimum stretch
which has not been satisfied.
\end{itemize}
\end{minipage}\\\\

\hline
\end{tabular}
\end{center}

\begin{theorem}
\label{thm:ssf-single-machine}
The algorithm $\sug$ is $(1 + \eps)$-speed $({1 \over \eps})$-competitive
for minimizing the maximum delay factor in unicast scheduling.
\end{theorem}
\begin{proof}
  Consider an arbitrary request sequence $\sigma$ and let $\alpha$ be
  the maximum delay factor achieved by $\sug$ on $\sigma$. If $\alpha
  = 1$ there is nothing to prove, so assume that $\alpha > 1$.  Let
  $J_i$ be the request that witnesses $\alpha$, that is $\alpha = (f_i
  - a_i)/S_i$. Note that $\sug$ does not process any request with slack
  more than $S_i$ in the interval $[a_i, f_i]$. Let $t$ be the largest
  value less than or equal to $a_i$ such that $\sug$ processed only
  requests with slack at most $S_i$ in the interval $[t, f_i]$. It
  follows that $\sug$ had no requests with slack $\le S_i$ just before
  $t$. The total work that $\sug$ processed in $[t,f_i]$ on requests
  with slack less than equal to $S_i$ is $(1+\eps)(f_i-t)$ and all
  these requests arrive in the interval $[t,f_i]$. An optimal offline
  algorithm with $1$-speed can do total work of at most $(f_i-t)$ in
  the interval $[t,f_i]$ and hence the earliest time by which it can
  finish these requests is $f_i + \eps (f_i-t) \ge f_i + \eps(f_i -
  a_i)$. Since all these requests have slack at most $S_i$ and have
  arrived before $f_i$, it follows that $\alpha^* \ge \eps(f_i -
  a_i)/S_i$ where $\alpha^*$ is the maximum delay factor of the
  optimal offline algorithm with $1$-speed machine. Therefore, we have
  that $\alpha/\alpha^* \le 1/\eps$.
\end{proof}

\begin{remark}
  For unit-time requests, the algorithm that {\em non-preemptively}
  schedules requests with the shortest slack is a $(1+\eps)$-speed $2
  \over \eps$-competitive for maximum delay factor.
\end{remark}

\subsection{Multiple Machine Scheduling}
\label{unicast-multi}

We now consider delay factor scheduling when there are $m$
machines. To adapt $\sug$ to this setting we take intuition from
previous work on minimizing $L_p$ norms of flow time and stretch
\cite{BansalP03,AvrahamiA03,ChekuriGKK04}. We develop an algorithm
that immediately dispatches an arriving request to a machine, and
further does not migrate an assigned request to a different
machine once it is assigned. Each machine essentially runs the
single machine $\sug$ algorithm and thus the only remaining
ingredient to describe is the dispatching rule. For this purpose
the algorithm groups requests into classes based on their slack. A
request $J_i$ is said to be in class $k$ if $S_i \in [2^k,
2^{k+1})$. The algorithm maintains the total processing time of
requests (referred to as {\em volume}) that have been assigned to
machine $x$ in each class $k$. Let $U^x_{=k}(t)$ denote the total
processing time of requests assigned to machine $x$ by time $t$ of class $k$.
With this notation, the algorithm $\mmug$ (for $\sug$ with
immediate dispatch) can be described.

\begin{center}
\begin{tabular}[r]{|c|}
\hline
\textbf{Algorithm}: \mmug \\

\\

\begin{minipage}{13cm}
\begin{itemize}
\item When a new request $J_i$ of class $k$ arrives at time $t$,
assign it to a machine $x$ where $U^x_{=k}(t) = \min_y U^y_{=k}(t)$.
\item Use $\sug$ on each machine separately.
\end{itemize}
\end{minipage}\\\\

\hline
\end{tabular}
\end{center}

The rest of this section is devoted to the proof of the following
theorem.
\begin{theorem}
\label{thm:multiple}
$\mmug$ is a $(1+ \eps)$-speed $O({1 \over
\eps})$-competitive algorithm for online delay factor
scheduling on $m$ machines.
\end{theorem}

We need a fair amount of notation. For each time $t$, machine $x$, and
class $k$ we define several quantities. For example $U^x_{=k}(t)$ is
the total volume assigned to machine $x$ in class $k$ by time $t$. We
use the predicate ``$\le k$'' to indicate classes $1$ to $k$.  Thus $U^x_{\le
  k}(t)$ is the total volume assigned to machine $x$ in classes $1$ to
$k$. We let $R^x_{=k}(t)$ to denote the remaining processing time on
machine $x$ at time $t$ and let $P^x_{=k}(t)$ denote the total volume
that $x$ has finished on requests in class $k$ by time $t$. Note that
$P^x_{=k}(t) = U^x_{=k}(t) - R^x_{=k}(t)$. All these quantities refer
to the algorithm $\mmug$. We use $V^*_{=k}(t)$ and $V_{=k}(t)$ to
denote the remaining volume of requests in class $k$ in an optimal
offline algorithm with speed $1$ and $\mmug$ with speed $(1+\eps)$,
respectively. Observe that $V_{=k}(t) = \sum_x R^x_{=k}(t)$.  The
quantities $V^*_{\le k}(t)$ and $V_{\le k}(t)$ are defined
analogously.

The algorithm $\mmug$ balances the amount of processing time for
requests with similar slack.  Note that the assignment of requests
is not based on the current volume of unfinished requests on the
machines, rather the assignment is based on the volume of requests
that were assigned in the past to different machines.  We begin
our proof by showing that the volume of processing time of
requests less than or equal to some slack class is almost
the same on the different machines at any time. Several of these
lemmas are essentially the same as in \cite{AvrahamiA03}.

\begin{observation}
\label{observation} For any time $t$ and two machines $x$ and $y$,
$|U^{x}_{=k}(t) - U^{x}_{=k}(t)| \leq 2^{k+1}$.  This also implies that
$|U^{x}_{\leq k}(t) - U^{x}_{\leq k}(t)| \leq 2^{k+2}$.
\end{observation}
\begin{proof}
  The first inequality holds since all of the requests of class $k$
  are of size $\leq 2^{k+1}$.  The second inequality follows easily
  form the first.
\end{proof}

Proofs of the next two lemmas can be found in the appendix.

\begin{lemma}
  \lemlab{processbound} Consider any two machines $x$ and $y$.  The
  difference in volume of requests that have already have been
  processed is bounded as $|P^{x}_{\leq k} (t) - P^{y}_{\leq k} (t)|
  \leq 2^{k+2}$.
\end{lemma}

\begin{lemma}
  \lemlab{closebound} At any time $t$ the difference between the
  residual volume of requests that needs to be processed, on any two
  different machines, $x$ and $y$ is bounded as $|R^{x}_{\leq k} (t) -
  R^{y}_{\leq k}(t)| \leq 2^{k+3}$.
\end{lemma}

\begin{corollary}
  \corlab{optarrival} At any time $t$, $V^*_{\leq k} (t) \geq V_{\leq
    k} (t) - m2^{k+3}$.
\end{corollary}


Now we get to the proof of the upper bound on \mmug, when given $(1+
\eps)$-speed, in a similar fashion to the single machine
case. Consider an arbitrary request sequence $\sigma$ and let $J_i$ be
the request that witnesses the delay factor $\alpha$ of $\mmug$ on $\sigma$.
Let $k$ be the class of $J_i$. Therefore $\alpha = (f_i - a_i)/S_i$.
Also, let $x$ be the machine on which $J_i$ was processed by $\mmug$.
We use $\alpha^*$ to denote the delay factor of some fixed optimal
offline algorithm that uses $m$ machines of speed $1$.

Let $t$ be the last time before $a_i$ when machine $x$ processed a request
of class $> k$. Note that $t \leq a_i$ since $x$ does not process any
request of class $> k$ in the interval $[a_i,f_i]$. At time $t$ we
know by \corref{optarrival} that $V^*_{\leq k} (t) \geq
V_{\leq k} (t) - m2^{k+3}$. If $f_i \le a_i + 2^{k+4}$ then $\mmug$
achieves a competitive ratio of $16$ since $J_i$ is in class
$k$. Thus we will assume from now on that $f_i > a_i + 2^{k+4}$.

In the interval $I = [t, f_i)$, $\mmug$ completes a total volume of of
$(1 + \eps)(f_i - t)$ on machine $x$. Using \lemref{processbound}, any
other machine $y$ also processes a volume of $(1+\eps)(f_i-t) -
2^{k+3}$ in $I$. Thus the total volume processed by $\mmug$ during $I$
in requests of classes $\le k$ is at least $m(1+\eps)(f_i-t) -
m2^{k+3}$.  During $I$, the optimal algorithm finishes at most
$m(f_i-t)$ volume in classes $\le k$.  Combining this with
\corref{optarrival}, we see that
\begin{eqnarray*}
V^*_{\le k}(f_i) & \ge & V_{\le k}(t) - m2^{k+3} + m(1+\eps)(f_i-t)  - m2^{k+3} \\
  & \ge & V_{\le k}(t) + m(1+\eps)(f_i-t) - m2^{k+4} \ge   \eps m (f_i - t).
\end{eqnarray*}
In the penultimate inequality we use the fact that $f_i-t \ge f_i -
a_i \ge 2^{k+4}$. Without loss of generality assume that no
requests arrive exactly at $f_i$. Therefore $V^*_{\le k}(f_i)$ is the
total volume of requests in classes $1$ to $k$ that the optimal
algorithm has left to finish at time $f_i$ and all these requests have
arrived before $f_i$. The earliest time that the optimal algorithm can
finish all these requests is by $f_i + \eps (f_i-t)$ and therefore it
follows that $\alpha^* \ge \eps (f_i - t)/2^{k+1}$. Since $\alpha \le (f_i-a_i)/2^k$ and $t \leq a_i$, it follows that $\alpha \le 2 \alpha^*/\eps$.

Thus $\alpha \le \max\{16, 2\alpha^*/\eps\}$ which finishes the proof
of Theorem~\ref{thm:multiple}.

\section{Broadcast Scheduling}
\label{sec:broadcast} We now move our attention to the broadcast
model where multiple requests can be satisfied by the transmission
of a single page.
Most of the literature in broadcast scheduling is concerned with the
case where all pages have the same size which is assumed to be unit. A
notable exception is the work of Edmonds and Pruhs
\cite{EdmondsP03}. Here we consider both the unit-sized as well as
arbitrary sized pages.

We start by showing that no $1$-speed online algorithm can be $(n/4)$-competitive for delay factor
where $n$ is the total number of unit-sized pages. We then show in
Section~\ref{broadcast-single} that there is a $(2+\eps)$-speed
$O(1/\eps^2)$-competitive algorithm for unit-sized pages. We prove
this for the single machine setting and it readily extends to the
multiple machine case. Finally,  we extend our algorithm and analysis to
the case of different page sizes to obtain a $(4+\eps)$-speed
$O(1/\eps^2)$-competitive algorithm in Section~\ref{subsec:varying}.
We believe that this can be extended to the multiple machine setting
but leave it for future work.

\begin{theorem}
  \label{broadcast-lower} Every $1$-speed online algorithm for broadcast
  scheduling to minimize the maximum delay factor is $\Omega(n)$-competitive
  where $n$ is number of unit-sized pages.
\end{theorem}

The proof of Theorem~\ref{broadcast-lower} can be found in the appendix.

\subsection{A Competitive Algorithm for Unit-sized Pages}
\label{broadcast-single}

We now develop an online algorithm, for unit-sized pages, that is
competitive given extra speed. It is easy to check that unlike in the
unicast setting, simple algorithms such as $\sug$ fail to be constant
competitive in the broadcast setting even with extra speed. The reason
for this is that any simple algorithm can be made to do an arbitrary
amount of ``extra'' work by repeatedly requesting the same page while
the adversary can wait and finish all these requests with a single
transmission. We use this intuition to develop a variant of $\sug$
that {\em adaptively} introduces waiting time for requests.  The
algorithm uses a single real-valued parameter $c < 1$ to control the
waiting period. The algorithm $\sbg$ ($\sug$ with waiting) is formally
defined below. We note that the algorithm is non-preemptive in that a
request once scheduled is not preempted. As we mentioned earlier, for
unit-sized requests, preemption is not very helpful. The algorithm
keeps track of the maximum delay factor it has seen so far,
$\alpha_t$, and this depends on requests that are yet to be completed
(we set $\alpha_0 = 1$). The important feature of the algorithm is
that it considers requests for scheduling only after they have waited
sufficiently long when compared to their {\em adaptive} slack.

\begin{center}
\begin{tabular}[r]{|c|}

\hline

\textbf{Algorithm}: \sbg \\

\\

\begin{minipage}{16cm}
\begin{itemize}
\item Let $\alpha_t$ be the maximum delay factor $\sbg$ has at
time $t$.

\item At time $t$,  let $Q(t) = \{ J_{(p,i)} \mid \mbox{ $J_{(p,i)}$ has not been satisfied and ${{t - a_{(p,i)} \over {\st_{p,i}}} } \geq c \alpha_t \st_{(p,i)}$} \}$.

\item If the machine is free at $t$, schedule the request in
$Q(t)$ with the smallest slack {\em non-preemptively}.
\end{itemize}
\end{minipage}\\

\hline
\end{tabular}
\end{center}
\vspace{4mm}

We now analyze $\sbg$ when it is given a $(2 + \eps)$-speed
machine. Let $\sigma$ be an arbitrary sequence of requests. Consider
the \emph{first} time $t$ where $\sbg$ achieves the maximum delay
factor $\alpha^{\sbg}$. At time $t$, $\sbg$ must have finished a
request $J_{(p,k)}$ which caused $\sbg$ to have this delay
factor. Hence, $\sbg$ has a maximum delay factor of $(f_{(p,k)} -
a_{(p,k)})/\st_{(p,k)}$ where $f_{(p,k)}$ is the time $\sbg$ satisfies
request $J_{(p,k)}$. We let $\opt$ denote some fixed offline optimum
algorithm and let $\alpha^*$ denote the optimum delay factor.

We now prove the most interesting difference between unicast and
broadcast scheduling. The following lemma shows that forcing a request
to wait in the queue, for a small period of time, can guarantee that
our algorithm is satisfying as many requests as $\opt$ by a single
broadcast unless $\opt$ has a similar delay factor.

Since $J_{(p,k)}$ defines $\alpha^{\sbg}$, we observe that from time
$t' = a_{(p,k)} + c(f_{(p,k)} - a_{(p,k)})$, the request
$J{(p,k)}$ is ready to be scheduled and hence the algorithm is
continuously busy in the interval $I = [t', f_{(p,k)}]$ processing
requests of slack no more than that of $J_{(p,k)}$.

\begin{lemma}
  \label{lem:main}
  Consider the interval $I = [t', f_{(p,k)})$. Suppose two distinct
  requests $J_{(x,j)}$ and $J_{(x,i)}$ for the same page $x$ were
  satisfied by $\sbg$ during $I$ at different times. If $\opt$
  satisfies both of these requests by a single broadcast then
  $\alpha^{\sbg} \le \frac{1}{c^2} \alpha^*$.
\end{lemma}

\begin{proof}
  Without loss of generality assume that $i >j$; therefore $a_{(x,j)}
  \le a_{(x,i)}$. Request $J_{(x,i)}$ must have arrived during $I$,
  otherwise $\sbg$ would have satisfied $J_{(x,i)}$ when it satisfied
  $J_{(x,j)}$. We observe that $\alpha_{t'} \ge \frac{c(f_{(p,k)} -
    a_{(p,k)})}{\st_{(p,k)}} \ge c \alpha^{\sbg}$ since $J_{(p,k)}$ is
  still alive at $t'$.

  Since $J_{(x,j)}$ was scheduled after $t'$, it follows that $\sbg$ would
  have made it wait at least $c \alpha_{t'} \st_{(x,j)}$ which implies
  that
  \[
  f_{(x,j)} \ge a_{(x,j)} + c \alpha_{t'} \st_{(x,j)}.
  \]
  Note that $\sbg$ satisfies $J_{(x,i)}$ by a separate broadcast from
  $J_{(x,j)}$ which implies that $a_{(x,i)} > f_{(x,j)}$. However, $\opt$
  satisfies both requests by the same transmission which implies
  that $\opt$ finishes $J_{(x,j)}$ no earlier than $a_{(x,i)}$. Therefore
  the delay factor of $\opt$ is at least the delay factor for $J_{(x,j)}$ in
  $\opt$ which implies that
\[
\alpha^* \ge \frac{a_{(x,i)} - a_{(x,j)}}{\st_{(x,j)}} \ge
\frac{f_{(x,j)} - a_{(x,j)}}{\st_{(x,j)}} \ge \frac{c \alpha_{t'}
  \st_{(x,j)}}{\st_{(x,j)}} \ge c \alpha_{t'} \ge c^2 \alpha^{\sbg}.
\]
\end{proof}

Note that previous lemma holds for any two requests scheduled by
$\sbg$ during interval $I$ regardless of when $\opt$ schedules them,
perhaps even after $f_{(p,k)}$.

\begin{lemma}
  \label{lem:arrivaltime}
  Consider the interval $I = [t', f_{(p,k)}]$. Any request which
  $\sbg$ scheduled during $I$ must have arrived after time $a_{(p,k)} -
  (1-c)(f_{(p,k)} - a_{(p,k)})$.
\end{lemma}

\begin{proof}
  For sake of contradiction, assume that a request $J_{(x,j)}$
  scheduled by $\sbg$ on the interval $I$ has arrival time less than
  $a_{(p,k)} - (1-c)(f_{(p,k)} - a_{(p,k)})$. Since $\sbg$ finishes
  this request during $I$, $f_{(x,j)} \ge a_{(p,k)} + c(f_{(p,k)} - a_{(p,k)})$.
  Also, as we observed before, all requests scheduled during $I$ by
  $\sbg$ have slack no more than that of $J_{(p,k)}$ which implies
  that $\st_{(x,j)} \le \st_{(p,k)}$. However this implies that
  the delay factor of $J_{(x,j)}$ is at least
    \begin{eqnarray*}
      \frac{(f_{(x,j)} - a_{(x,j)})}{\st_{(x,j)}} & \ge & \frac{a_{(p,k)} + c(f_{(p,k)} - a_{(p,k)}) - (a_{(p,k)} - (1-c)(f_{(p,k)} - a_{(p,k)}))}{\st_{(x,j)}} \\
        & \ge & \frac{(f_{(p,k)} - a_{(p,k)})}{\st_{(x,j)}}
        ~ \ge ~ \frac{(f_{(p,k)} - a_{(p,k)})}{\st_{(p,k)}} ~ \ge ~ \alpha^{\sbg}.
      \end{eqnarray*}

      This is a contradiction to the fact that $J_{(p,k)}$ is the
      first request that witnessed the maximum delay factor of \sbg.
\end{proof}

\medskip

Now we are ready to prove the competitiveness of $\sbg$.

\begin{lemma}
  \lemlab{eps-ceps-c-competitive} The algorithm $\sbg$ when given a
  $(2+\eps)$-speed machines satisfies $\alpha^{\sbg} \le
  \max\{\frac{1}{c^2}, \frac{1}{\eps - c\eps - c}\} \alpha^*$.
\end{lemma}

\begin{proof}
The number of broadcasts which $\sbg$ transmits during the
interval $I = [t', f_{(p,k)}]$ is
\[
(2 + \epsilon)(f_{(p,k)} - t') \ge (2+\eps)(1-c)(f_{(p,k)} - a_{(p,k)}).
\]

From Lemma~\ref{lem:arrivaltime}, all the requests processed during
$I$ have arrived no earlier than $a_{(p,k)} - (1-c)(f_{(p,k)} -
a_{(p,k)})$. Also, each of these requests has slack no more than
$\st_{(p,k)}$. We restrict attention to the requests satisfied by $\sbg$
during $I$. We consider two cases

First, if there are two requests for the same page that $\sbg$ satisfies
via distinct broadcasts but $\opt$ satisfies using one broadcast, then
by Lemma~\ref{lem:main}, $\alpha^{\sbg} \le \frac{1}{c^2}\alpha^*$ and we
are done.

Second, we assume that $\opt$ does not merge two requests for the same
page whenever $\sbg$ does not do so. It follows that $\opt$ also has
to broadcast $(2+\eps)(1-c)(f_{(p,k)} - a_{(p,k)})$ pages to satisfy
the requests that $\sbg$ did during $I$. Since these
requests arrived no earlier than $a_{(p,k)} - (1-c)(f_{(p,k)} - a_{(p,k)})$,
$\opt$, which has a $1$-speed machine, can finish them at the earliest
by
\[
(2+\eps)(1-c)(f_{(p,k)} - a_{(p,k)}) + a_{(p,k)} - (1-c)(f_{(p,k)} - a_{(p,k)}) \ge f_{(p,k)} + (\eps - c - c\eps)(f_{(p,k)} - a_{(p,k)}).
\]

Since each of these requests has slack at most $\st_{(p,k)}$ and arrived
no later than $f_{(p,k)}$, we have that
\begin{eqnarray*}
  \alpha^* & \ge & (f_{(p,k)} + (\eps - c - c\eps)(f_{(p,k)} - a_{(p,k)}) - f_{(p,k)})/\st_{(p,k)} \\
  & \ge & (\eps - c - c\eps)(f_{(p,k)} - a_{(p,k)})/\st_{(p,k)}
  ~ \ge ~  (\eps - c - c\eps)\alpha^{\sbg}.
\end{eqnarray*}
\end{proof}

The previous lemma yields the following theorem.

\begin{theorem}
  With $c = \eps/2$, $\sbg$ is a $(2+ \eps)$-speed $O({1 \over
    {\eps^2}})$-competitive algorithm for minimizing the maximum delay
  factor in broadcast scheduling with unit-sized pages.
\end{theorem}

It may appear that $\sbg$ needs knowledge of $\eps$. However, another
way to interpret \lemref{eps-ceps-c-competitive} is that for any fixed
constant $c$, $\sbg$ with parameter $c$ is constant competitive in all
settings where its machine is at least $(2+2\sqrt{c})$ times the speed
of the optimal algorithm. Of course, it would be ideal to have an
algorithm scales with $\eps$ without any knowledge of $\eps$. We leave
the existence of such an algorithm for future work.

Now consider having $m$ machines where we have $(1+ \eps)$-speed.
Since we are using unit time requests, this is analogous to $\opt$
having one $m$-speed machine and $\sbg$ having a $(m(2 +
\eps))$-speed machine.  Thus, one can extend the above analysis to
the multiple machine setting with unit-sized pages in a
straight forward fashion.

\subsection{Varying Page Sizes}
\label{subsec:varying}
In this section we generalize our algorithm for unit-sized pages to
the setting where each page has potentially a different page size. We
let $\ell_p$ denote the length of page $p$. In this setting we allow
preemption of transmissions. Suppose the transmission of a page $p$ is
started at time $t_1$ and ends at time $t_2$; $p$ may be preempted for
other transmissions and hence $t_2-t_1 \ge p$.  A request for a page
$p$ is satisfied by the transmission of $p$ during the interval $[t_1,
t_2]$ only if the request arrives before $t_1$. It is possible that
the transmission of a page $p$ is abandoned and {\em restarted} due to
the arrival of a new request for $p$ with a smaller slack.  This may
lead to further wasted work by the algorithm and increases the
complexity of the analysis. Here we show that a natural adaptation of
$\sbg$ is competitive even in this more general setting if it is given $(4 + \eps)$-speed.

We outline the details of modifications to $\sbg$. As before, at any time $t$, the algorithm considers broadcasting a
request $J_{(p,i)}$ if ${{t - a_{(p,i)}}\over{\st_{(p,i)}}} \geq c
\alpha_t \st_{(p,i)}$; these are requests that have waited long
enough.  Among these requests, the one with the smallest slack is
scheduled. Note that the waiting is only for requests that have not
yet been started; any request that has already started transmission is
available to be scheduled.  The algorithm breaks ties arbitrarily, yet
ensures that if a request $J_{(p,k)}$ is started before a request
$J_{(p',j)}$ then $J_{(p,k)}$ will be finished before request
$J_{(p',j)}$. Note that the algorithm may preempt a request
$J_{(p,i)}$ by another request $J_{(p,k)}$ for the same page $p$ even
though $i < k$ if $\st_{(p,k)} < \st_{(p,i)}$. In this case the
transmission of $J_{(p,i)}$ is effectively abandoned. Note that
transmission of a page $p$ may be repeatedly abandoned.

We now analyze the algorithm assuming that it has a $(4+\eps)$-speed
advantage over the optimal offline algorithm. The extra factor in
speed is needed in our analysis to handle the extra wasted work due to
potential retransmission of a page $p$ after a large portion of it has
already been transmitted.  As before, let $\sigma$ be a sequence of
requests and let $t$ be the \emph{first} time $\sbg$ achieves the
maximum delay factor $\alpha^{\sbg}$.  At time $t$, it must be the
case that a request $J_{(p,k)}$ was finished which caused $\sbg$ to
have his maximum delay factor.  Hence, $\sbg$ has a maximum delay
factor of $(f_{(p,k)} - a_{(p,k)}/ \st_{(p,k)})$ where $f_{(p,k)}$ is
the time $\sbg$ satisfied request $J_{(p,k)}$.

As with the case with unit time requests, at time $t' = a_{(p,k)}
+ c(f_{(p,k)} - a_{(p,k)})$ the request $J_{(p,k)}$ is ready to be
scheduled and the algorithm is busy on the interval $I=[t',
f_{(p,k)}]$ processing requests of slack at most $\st_{(p,k)}$.

We say that a request $J_{(p,i)}$ is {\em started} at time $t$ if $t$
is the first time at which the algorithm picked $J_{(p,i)}$ to
transmit its page. Multiple requests may be waiting for the same page
$p$ but only the request with the smallest slack that is picked by the
algorithm is said to be started. Thus a request may be satisfied
although it is technically not started. Also, a request $J_{(p,i)}$
that is started may be abandoned by the start of another request for
the same page.

The lemma below is analogous to Lemma~\ref{lem:main} but requires
a more careful statement since requests may now be started and
abandoned.

\begin{lemma}
  \label{mainvarying} Consider two distinct requests $J_{(x,j)}$ and
  $J_{(x,i)}$ for the same page $x$ where $i>j$ such that they are
  both satisfied by $\opt$ via the same transmission. If $\sbg$ starts
  $J_{(x,j)}$ in $[t', f_{(p,k)}]$ before the arrival of $J_{(x,i)}$,
  then  $\alpha^{\sbg} \leq {1 \over c^2}\alpha^*$.
\end{lemma}

Observe that the request $J_{(x,j)}$ may be satisfied together with
$J_{(x,i)}$ even though it starts before the arrival of $J_{(x,i)}$.

\begin{proof}
  As before, $\alpha_{t'} \ge \frac{c(f_{(p,k)} -
    a_{(p,k)})}{\st_{(p,k)}} \ge c \alpha^{\sbg}$ since $J_{(p,k)}$ is
  still alive at $t'$. Since $J_{(x,j)}$ is started after $t'$, it
  follows that $\sbg$ would have made it wait at least $c \alpha_{t'}
  \st_{(x,j)}$. Let $t \ge t'$ be the start time of
  $J_{(x,j)}$. Therefore $t \ge a_{(x,j)} + c \alpha_{t'}
  \st_{(x,j)}$. By our assumption, $t < a_{(x,i)}$ and therefore
  $ a_{(x,i)} > a_{(x,j)} + c \alpha_{t'}\st_{(x,j)}$.

  Since $\opt$ satisfies these two requests by the same transmission,
  the finish time of $J_{(x,j)}$ in $\opt$ is at least $a_{(x,i)}$.
  Therefore,

\[
\alpha^* \ge \frac{a_{(x,i)} - a_{(x,j)}}{\st_{(x,j)}} \ge \frac{c
\alpha_{t'}  \st_{(x,j)}}{\st_{(x,j)}} \ge c \alpha_{t'} \ge c^2
\alpha^{\sbg}.
\]
\end{proof}

The proof of the lemma below is very similar to that
of Lemma~\ref{lem:arrivaltime}.
\begin{lemma}
  \label{lem:arrivaltimevarying}
  Consider the interval $I = [t', f_{(p,k)}]$. Any request which is
  alive with slack $\leq S_{(p,k)}$, but unsatisfied by $\sbg$ at time $t'$ must have arrived after
  time $a_{(p,k)} - (1-c)(f_{(p,k)} - a_{(p,k)})$.
\end{lemma}

Now we are ready to prove the competitiveness of $\sbg$.
Although the outline of the proof is similar to that of
\lemref{eps-ceps-c-competitive}, it requires more careful reasoning
to handle the impact of abandoned transmissions
of pages. Here is where we crucially rely on the speed of $(4+\eps)$.

\begin{lemma}
  \lemlab{eps-ceps-c-competitivevarying}
  The algorithm $\sbg$ when given a $(4+\eps)$-speed machine
  satisfies $\alpha^{\sbg} \le \max\{\frac{1}{c^2}, \frac{2}{(\eps - c\eps - c)}\} \alpha^*$.
\end{lemma}

\begin{proof}
  We consider the set of requests satisfied by $\sbg$ during the
  interval $I = [t', f_{(p,k)}]$. All of these requests have slack at
  most $\st_{(p,k)}$, and from Lemma~\ref{lem:arrivaltimevarying} and
  the property of the algorithm, have arrived no earlier than
  $a_{(p,k)} - (1-c)(f_{(p,k)} - a_{(p,k)})$. Since $\sbg$ is busy throughout
  $I$, the volume of broadcasts it transmits during $I$
  is $(4 + \epsilon)(f_{(p,k)} - t') \ge (4+\eps)(1-c)(f_{(p,k)} -
  a_{(p,k)})$.

  We now argue that either Lemma~\ref{mainvarying} applies in which
  case $\alpha^{\sbg} \le \alpha^*/c^2$, or $\opt$ has to transmit
  a comparable amount of volume to that of $\sbg$.

  Fix a page $x$ and consider the transmissions for $x$ that $\sbg$
  does during $I$. Let $J_{(x,i_1)}, J_{(x,i_2)}, \ldots, J_{(x,i_r)}$
  be distinct requests for $x$ which cause these
  transmissions. Amongst these, only $J_{(x,i_1)}$ may have started
  before $t'$, the rest start during $I$.  Note that we are not
  claiming that these transmissions are satisfied separately; some of
  them may be preempted and then satisfied together. Observe that if
  $J_{(x,i_h)}$ starts at some time $t$ then it implies that no
  request $J_{(x,i_{h'})}$ for $h' > h$ has arrived by time $t$.
  Therefore by Lemma~\ref{mainvarying}, if $\opt$ satisfies any two of
  these requests that $\sbg$ started in $I$ by the same transmission,
  $\alpha^{\sbg} \le \alpha^*/c^2$ and we are done.

  Otherwise, $\opt$ satisfies each of $J_{(x,i_2)}, \ldots,
  J_{(x,i_r)}$ by separate transmissions. (If $J_{(x,i_1)}$ was
  started by $\sbg$ before $t'$, $\opt$ could satisfy $J_{(x,i_1)}$
  and $J_{(x,i_2)}$ together and we would not be able to invoke
  Lemma~\ref{mainvarying}). Therefore if $r \ge 2$ then the total
  volume of transmissions that $\opt$ does to satisfy these requests
  for page $x$ is at least $(r-1)\ell_x$ while $\sbg$ does at most
  $r\ell_x$.  If $r = 1$ then both $\opt$ and $\sbg$ transmit page $x$ once for its entire page length. In either case, the total volume of transmissions
  that $\opt$ does is at least half those of $\sbg$. Since $x$ was
  arbitrary, it follows that the total number of transmissions that
  $\opt$ does to satisfy requests that $\sbg$ satisfies during $I$
  is at least $\frac{1}{2} (4+\eps)(1-c)(f_{(p,k)} -  a_{(p,k)})$.

  From Lemma~\ref{lem:arrivaltimevarying}, all the requests that
  $\sbg$ processes during $I$ arrived no earlier than $a_{(p,k)} -
  (1-c)(f_{(p,k)} - a_{(p,k)})$. Since $\opt$ has a $1$-speed machine,
  it follows that $\opt$ can finish these requests only by time
  \[
  {1 \over 2}(4+\eps)(1-c)(f_{(p,k)} - a_{(p,k)}) + a_{(p,k)} -
  (1-c)(f_{(p,k)} - a_{(p,k)}) \ge f_{(p,k)} + {1 \over 2}(\eps - c
  - c\eps)(f_{(p,k)} - a_{(p,k)}).
  \]
  Since each of these requests have slack at most $\st_{(p,k)}$ and arrive
  no later than $f_{(p,k)}$,

\[
  \alpha^*  \ge (f_{(p,k)} + {1 \over 2}(\eps - c - c\eps)(f_{(p,k)} - a_{(p,k)}) - f_{(p,k)})/\st_{(p,k)}  \ge  {1 \over 2}(\eps - c - c\eps) (f_{(p,k)} - a_{(p,k)})/\st_{(p,k)} \ge  {1 \over 2}(\eps - c - c\eps)\alpha^{\sbg}.
\]
\end{proof}

We thus obtain the following.

\begin{theorem}
  With $c = \eps/2$, $\sbg$ is a $(4+ \eps)$-speed $O({1 \over
    {\eps^2}})$-competitive algorithm for minimizing the maximum delay
  factor in broadcast scheduling with arbitrary page sizes.
\end{theorem}

\section{Concluding Remarks}
In this paper we have initiated the study of online algorithms for
minimizing delay factor when requests have deadlines.  Our main result
is broadcast scheduling where the algorithm and analysis demonstrates
the utility of making requests wait. We hope that this and related
ideas are helpful in understanding other performance measures in the
broadcast setting. Particularly, can `waiting' combined with some
known algorithm, like most requests first, be used to improve the
current best known online algorithm for minimizing the average
response time? Another interesting problem is whether there is a
$(1+\eps)$-speed $O(1)$-competitive algorithm for delay factor.  Our
algorithm has a parameter that controls the waiting time. Is there an
algorithm that avoids taking an explicit parameter and ``learns'' it
along the way?

\bigskip
\noindent
{\bf Acknowledgments:} We thank Samir Khuller for
clarifications on previous work and for his encouragement.

\bibliographystyle{alpha}
\bibliography{DelayFactor}

\appendix

\section{Omitted Proofs}

\subsection{Proof of Theorem~\ref{thm:unicast_lb}}

\begin{proof}
For sake of contradiction, assume that some algorithm that
achieves a competitive ratio better than  ${{P^{.4}} \over
2}$ exists. Now consider the following example.

\begin{quotation}
\noindent \texttt{Type 1}:  At time $0$ let the client request a
page with processing time and deadline $P$.  This request
has slack $P$.
\end{quotation}

\begin{quotation}
\noindent \texttt{Type 2}:  At times $P - P^{.6},
P, P + P^{.6}, \ldots, P^{1.16} - P^{.6}$
let the client request a page with processing time $P^{.6}$
and a deadline $P^{.6}$ time units after its arrival time.
These requests have slack $P^{.6}$. \noindent
\end{quotation}

Consider time $P^{1.16}.$  Assume that this is all of the
requests which the client makes.  The optimal solution schedules
these requests in a first in first out fashion.  The optimal
schedule finishes request type 1 by its deadline.  The requests of
type 2 then finish at $P^{.6}$ time units after their
deadline. Thus, the delay factor for the optimal schedule is
$2P^{.6}/P^{.6} = 2$.

The maximum ratio of maximum to minimum slack values seen so far
is ${P \over P^{.6}} = P^{.4}$.  Thus, the maximum
delay factor our algorithm can have is $(P^{.4})^{.4} / 2 =
P^{.16}/2$. Consider having the request of type 1 still in
the deterministic algorithms queue.  At time $P^{1.16}$, the
algorithm has achieved a delay factor of at least ${P^{1.16}
\over P} = P^{.16}$.  Thus, the algorithm has a
competitive ratio of at least $P^{.16} \over 2$, a
contradiction. Therefore, at time $P^{1.16}$ the algorithm
must have finished the request of type 1.  Now, immediately after
this time, requests of type 3 arrive.

\begin{quotation}
\noindent \texttt{Type 3}: Starting at time $P^{1.16}$ the
client requests $P^{1.2} - P^{.6}$ unit processing time
requests each with a deadline one time unit after their arrival
time. These requests arrive one after another, each time unit. The
slack of these requests is 1.
\end{quotation}

These are all of the requests which are sent.  The optimal
solution schedules the request of type 1 until time $P^{.4}$,
thus has $P^{.6}$ processing time left to finish this
request. Then the optimal solution schedules the type 2 and type 3
requests as they arrive, giving them a delay factor of 1.  At time
$P^{1.16} + P^{1.2} - P^{.6}$ the optimal solution
schedules the request of type 1 to completion.  Thus delay factor
of this solution is ${{P^{1.16} + P^{1.2}} \over
{P}} \leq 2P^{.2}$.

Our algorithm must have scheduled the request of type 1 by time
$P^{1.16}$.  Thus the last request it finishes is either of
type 2 or type 3.  If the request is of type 2 then this request
must have waited for all requests of type 3 to finish along with
its processing time, thus the delay factor is at least
${{P^{1.2} + P^{.6}} \over {P^{.6}}} \geq
P^{.6}$.  If the last request satisfied by the algorithm is
of type 3, then this request must have waited for a request of
type 2 to finish, so the delay factor is at least $P^{.6}$.
In either case, the competitive ratio of the algorithm is at least
${{P^{.6}} \over {2 P^{.2}}} = {{P^{.4}} \over
{2}}$, a contradiction.

\end{proof}

\subsection{Proof of \lemref{processbound}}
\begin{proof}
  Suppose the lemma is false. Then there is a first time $t_0$ when
  $P^{x}_{\leq k} (t_0) - P^{y}_{\leq k} (t_0) = 2^{k+2}$ and small
  constant $\delta t > 0$ such that $P^{x}_{\leq k} (t_0 + \delta t) -
  P^{y}_{\leq k} (t_0 + \delta t) > 2^{k+2}$. Let $t' = t_0+\delta t$.
  For this to occur, $x$ processes a request of class $\le k$ during
  the interval $I = [t_0, t']$ while $y$ processes a request of class
  $> k$. Since each machine uses $\sug$, it must be that $y$ had no
  requests in classes $\le k$ during $I$ which implies that $U^y_{\le
    k}(t') = P^y_{\le k}(t')$. Therefore,

\[
U^{y}_{\leq k} (t') = P^{y}_{\leq k} (t') < P^{x}_{\leq k} (t') - 2^{k+2}
\le U^{x}_{\leq k} (t') - 2^{k+2},
\]
since $P^{x}_{\leq k} (t') \le U^{x}_{\leq k} (t')$. However, this implies
that
\[
U^{y}_{\leq k} (t') < U^{x}_{\leq k} (t')  - 2^{k+2},
\]
a contradiction to Observation \ref{observation}.
\end{proof}

\subsection{Proof of \lemref{closebound}}

\begin{proof}
Combining Observation \ref{observation}, \lemref{processbound},
and the fact that $R(t) = U(t) - P(t)$ by definition then,

\[
|R^{x}_{\leq k} (t) -  R^{y}_{\leq k}(t)| \le
|U^{x}_{\leq k} (t) -  U^{y}_{\leq k}(t)| + |P^{x}_{\leq k} (t) -  P^{y}_{\leq k}(t)| \le 2^{k+3}.
\]
\end{proof}

\subsection{Proof of Theorem~\ref{broadcast-lower}}
\begin{proof}
  Let $A$ be any online $1$-speed algorithm. We consider the following
  adversary. At time $0$, the adversary requests pages $1, \ldots, {n
    \over 2}$, all which have a deadline of $n \over 2$. Between time
  $1$ and $n \over 4$ the client requests whatever page the online
  algorithm $A$ broadcasts immediately after that request is
  broadcast; this new request also has a deadline of $n \over 2$. It
  follows that at time $t = {n \over 2}$ the online algorithm $A$ has
  $n \over 4$ requests for distinct pages in its queue. However, the
  adversary can finish all these requests by time $n \over 2$. Then
  starting at time ${n \over 2}$ the adversary requests $n \over 2$
  new pages, say ${n \over 2} +1, \ldots, n$.  These new pages are
  requested, one at each time step, in a cyclic fashion for $n^2$
  cycles.  More formally, for $i = 1, \ldots, n/2$, page ${n \over 2}
  + i$ is requested at times $j\cdot({n \over 2}) +i - 1$ for $j = 1,
  \ldots, n$. Each of these requests has a slack of one which means
  that their deadline is one unit after their arrival.  The adversary
  can satisfy these requests with delay since it has no queue at any
  time; thus its maximum delay factor is $1$.  However, the online
  algorithm $A$ has $n \over 4$ requests in its queue at time $n \over
  2$; each of these has a slack of $n \over 2$. We now argue that the
  delay factor of $A$ is $\Omega(n)$. If the algorithm satisfies two
  slack $1$ requests for the same page by a single transmission, then
  its delay factor is $n/2$; this follows since the requests for the
  same page are $n/2$ time units apart. Otherwise, the algorithm does
  not merge any requests for the same page and hence finishes the
  the last request by time $n/2 + n^2/2 + n/4$. If the
  last request to be finished is a slack $1$ request, then its delay
  factor is at least $n/4$ since the last slack $1$ requests is released
  at time $n/2 + n^2/2$. If the last request to be finished is
  one of the requests with slack $n/2$, then its delay factor
  is at least $n^2/2/(n/2) = \Omega(n)$.
\end{proof}
\end{document}